\documentclass[a4paper,11pt]{article}
\usepackage{a4wide}
\usepackage{amsmath}
\usepackage{amsfonts}
\usepackage{amssymb}
\usepackage{amsthm}
\usepackage[utf8]{inputenc}
\usepackage{cellspace}
\usepackage[small,sc]{caption}
\usepackage{booktabs, multicol, multirow}
\usepackage{enumitem}

\theoremstyle{plain}
\newtheorem{theorem}{Theorem}

\newtheorem{lemma}{Lemma}

\theoremstyle{definition}

\setlist[1]{itemsep=-4pt}

\begin{document}

\title{Time-space tradeoffs for two-way finite automata}

\author{
	Shenggen Zheng\\
	\small \texttt{zhengshenggen@gmail.com}\\
    \small  Sun Yat-sen University\\
	\and
   Daowen Qiu\thanks{Corresponding author: Daowen Qiu.}\\
	\small \texttt{issqdw@mail.sysu.edu.cn}\\
    \small  Sun Yat-sen University\\
\and
	Jozef Gruska\\
	\small \texttt{gruska@fi.muni.cz}\\
	\small  Masaryk University
}

\renewcommand\footnotemark{}

\date{\vspace{-5ex}}

\maketitle

\begin{abstract}
We  explore bounds of  {\em time-space tradeoffs}
 in language recognition on {\em two-way finite automata} for some special languages. We  prove:  (1) a time-space tradeoff  upper bound for recognition of the languages $L_{EQ}(n)$ on {\em two-way probabilistic finite automata} (2PFA): $TS={\bf O}(n\log n)$, whereas a time-space tradeoff lower bound on  {\em two-way deterministic finite automata} is ${\bf \Omega}(n^2)$; (2) a time-space tradeoff  upper bound for recognition of the languages $L_{INT}(n)$ on {\em two-way finite automata with quantum and classical states} (2QCFA): $TS={\bf O}(n^{3/2}\log n)$, whereas a   lower bound on  2PFA is $TS={\bf \Omega}(n^2)$; (3) a time-space tradeoff  upper bound for recognition of the languages $L_{NE}(n)$ on  exact 2QCFA: $TS={\bf O}(n^{1.87} \log n)$, whereas a   lower bound on  2PFA is $TS={\bf \Omega}(n^2)$.

It has been proved (Klauck, STOC'00) that the exact one-way quantum finite automata have no advantage comparing to classical finite automata in recognizing languages. However, the result (3) shows that  the exact 2QCFA do have an advantage in comparison with  their classical counterparts, which has been the first example showing that the exact quantum computing have advantage in time-space tradeoff comparing to classical computing.

Usually, two communicating parties, Alice and Bob, are supposed to have an access to arbitrary computational power in {\em communication complexity} model that is used.
Instead of that we will consider communication complexity in such a setting that two parties are using only finite automata and we prove in this setting that quantum automata are better than classical automata and also probabilistic automata are better than deterministic automata for some well known tasks.

\end{abstract}
\section{Introduction}
Time-space tradeoffs is an important research topic in the study of complexity  of both classical and quantum computing \cite{BC82,Klk03,KSW07} with respect to various computing models \cite{BST98,BNS92,Cbm66}.  However, in the case of two-way finite automata, there is few work on their time-space tradeoffs. Mostly only their time complexity or state complexity (space complexity) has been investigated.

When just  time complexity or  state complexity (space complexity) of two-way finite automata to recognize some languages are considered,  it seems that quantum finite automata have no advantages at all compared to their classical counterparts. However, quite surprisingly,  when  time-space product is considered, then advantages of quantum variations of the  classical models can be demonstrated as shown in this paper.

Time-space tradeoffs are closely related to communication complexity. In this paper, we will use communication complexity results to derive time-space tradeoffs results for two-way finite automata. We  prove that the time-space tradeoffs for recognizing some languages in  two-way finite automata with quantum and classical states (2QCFA) \cite{Amb02} are better than in their classical counterparts and also that probabilistic two-way quantum finite automata (2PFA) \cite{Frd81} are better than  two-way deterministic finite automata (2DFA) \cite{HU79}.

Since the topic of communication complexity was introduced by Yao  \cite{Yao79}, it has been  extensively studied \cite{Buh09,KusNis97}.
In the setting of two  parties, Alice is given  an $x\in\{0,1\}^n$, Bob is given a $y\in\{0,1\}^n$ and their task is to communicate in order to determine the value of some given Boolean function $f:\{0,1\}^n\times\{0,1\}^n\to\{0,1\}$, while exchanging as small  number of bits  as possible. In this setting, local computations of the parties
are considered to be free, but communication is considered to be  expensive and has to be minimized.
Two of the most often studied communication problems  are  that of equality and intersection \cite{KusNis97}, defined as follows:
(1) {\bf Equality}: $\mbox{EQ}(x,y)=1$ if $x=y$ and 0 otherwise.
(2) {\bf Intersection}: $\mbox{INT}(x,y)=1$  if there is an index $i$ such that $x_i=y_i=1$ and $0$ otherwise.

\subsection{Time-space tradeoffs}\label{sub-sec}
Let us consider the following language over the alphabet $\Sigma=\{0,1,\#\}$:
\begin{equation}
    L_{EQ}(n)=\{x\#^ny\mid x,y\in\{0,1\}^n, \mbox{EQ}(x,y)=1\}.
\end{equation}
It is clear that 2DFA (therefore also 2PFA) can recognize $L_{EQ}(n)$. The time complexity\footnote{When two-way finite automata are used to recognize languages, they can halt before reading all the input.} of 2DFA recognizing this language is ${\bf O}(n)$. The state complexity of 2DFA  recognizing the language is ${\bf O}(n^2)$, that is the space used is ${\bf O}(\log n)$. The time complexity and also the space used of  2PFA  recognizing the same language  is almost the same.  However, when we consider time-space tradeoff for the language  $L_{EQ}(n)$,  the situation is very different.

 We will use a 2PFA to simulate the probabilistic communication protocol from Chapter 1 of \cite{Hrc05} for the  problem $EQ$ and get an  upper bound for the  time-space tradeoff for 2PFA.

\begin{theorem}\label{Th1}
There is a 2PFA  that accepts the language $L_{EQ}(n)$ in the time $T$ using the space $S$ such that
  $TS={\bf O}(n\log n)$.
  \end{theorem}

 Using communication complexity lower bound proof method \cite{KusNis97},   we can get the lower bound for time-space tradeoff for 2DFA.
  \begin{theorem}\label{Th2}
Let ${\cal A}$ be a 2DFA that accepts the language $L_{EQ}(n)$ in time $T$ using space $S$.
Then, $TS={\bf \Omega}(n^2)$.
\end{theorem}

In order to prove the time-space tradeoffs advantages of 2QCFA compared to 2PFA,
let us consider the following language over the  alphabet $\Sigma=\{0,1,\#\}$:
\begin{equation}
    L_{INT}(n)=\{x\#^ny\mid x,y\in\{0,1\}^n, \mbox{INT}(x,y)=1\}.
\end{equation}

We use a 2QCFA to simulate the quantum communication protocol from \cite{Buh98,Buh09} for the problem $\mbox{INT}$ and get an  upper bound for the  time-space tradeoff for 2QCFA.

\begin{theorem}\label{Th3}
There is a 2QCFA that accepts the language $L_{INT}(n)$ in time $T$ using space $S$ such that
  $TS={\bf O}(n^{3/2}\log n)$.
  \end{theorem}

  Buhrman et al.~\cite{Buh98} reduced certain quantum communication tasks  to computation problems, which is essentially a way to  transform quantum query algorithms
 to quantum communication protocols. More exactly, they showed that if there is a $t$-query quantum algorithm computing an $n$-bit Boolean function $f$ with an error $\varepsilon$, then there is a communication  protocol with ${\bf O}({t\log n})$ communication for the function $f(x\wedge y)$ with the same  error $\varepsilon$.

 The main idea in the   proofs of our main results  is to transform   quantum query algorithms
 and quantum communication protocols to algorithms for 2QCFA.

   Using one of communication complexity lower bound proof methods,   we can get the following lower bound for the time-space tradeoff for the language $L_{INT}(n)$  on 2PFA.
   \begin{theorem}\label{Th4}
Let ${\cal A}$ be a 2PFA that accepts the language $L_{INT}(n)$ in time $T$ using space $S$.
Then, $TS={\bf \Omega}(n^2)$.
\end{theorem}

Concerning the exact computing mode, Klauck \cite{Kla00} proved, for any regular language $L$,  that the state complexity  of the exact one-way quantum finite automata (1QFA) for $L$ is not less than the state complexity of an equivalent one-way deterministic finite automata (DFA). That means that the exact 1QFA have no advantage in recognizing regular languages.
It is therefore of interest to consider the case of  two-way finite automata.
 We still do not know whether there is time complexity or  state complexity advantages for two-way quantum finite automata in  recognition of languages.  However, we prove that exact 2QCFA do have  time-space tradeoff  advantages for recognizing some special languages.

Let us consider the sequence of functions studied in \cite{Amb13}.
We define  the
function $NE(x_1, x_2, x_3)$ as follows:
 $\mbox{NE}(x_1, x_2, x_3) = 0$ if $x_1 = x_2 = x_3$ and
   $\mbox{NE}(x_1, x_2, x_3) = 1$ otherwise.
Define
\begin{itemize}
  \item $\mbox{NE}^0(x_1)=x_1$ and
  \item $
     \mbox{NE}^d(x_1,\dots,x_{3^d})=\mbox{NE}(  \mbox{NE}^{d-1}(x_1,\ldots,x_{3^{d-1}}), \linebreak[0]  \mbox{NE}^{d-1}(x_{3^{d-1}+1},\linebreak[0]\ldots,x_{2\cdot3^{d-1}}), \linebreak[0] \mbox{NE}^{d-1} \linebreak[0] (x_{2\cdot3^{d-1}+1}, \linebreak[0] \ldots,x_{3^{d}}))
   $  for all $d>0$.
\end{itemize}

Let $n=3^d$. We define
$
\mbox{RNE}(x,y)=\mbox{NE}^d(x_1\wedge y_1,\ldots,x_n\wedge y_n),
 $
   where $x,y\in\{0,1\}^n$, and let us
 consider the following language
\begin{equation}
    L_{NE}(n)=\{x\#^ny\mid x,y\in\{0,1\}^n, \mbox{RNE}(x,y)=1\}.
\end{equation}

We will use a 2QCFA to simulate the quantum communication protocol from \cite{Amb13} for the problem $\mbox{RNE}$ and get an  upper bound for the  time-space tradeoff for 2QCFA.
\begin{theorem}\label{Th5}
There is an exact 2QCFA that accepts the language $L_{NE}(n)$ in time $T$ using space $S$ such that
  $TS={\bf O}(n^{1.87} \log n)$.
  \end{theorem}
   \begin{theorem}\label{Th6}
Let ${\cal A}$ be a 2PFA that accepts the language $L_{NE}(n)$ in time $T$ using space $S$.
Then, $TS={\bf \Omega}(n^2)$.
\end{theorem}

\subsection{Communication  of finite automata}
Two communicating parties Alice and Bob are usually supposed to have unlimited computational power in communication complexity models.
However we will consider a very different setting. Namely that two parties are using only finite automata  for their internal computation. In this setting, Alice and Bob will be sending only some  states of their finite automata as messages to  each other.
At the beginning, Alice does some computation on her finite automaton, then sends a state $s$ of her  automaton to Bob.  After receiving the state $s$ from Alice, Bob does  computation with $s$ as  the starting state  on his automaton and then after some computation Bob sends his state $t$ to Alice.  Alice then resumes computation in her automaton  with the starting  state $t$ and so on.
In case Alice and Bob are using 2QCFA, they can send both quantum and classical states.

We prove that the communication complexity for  problems $\mbox{EQ}$, $\mbox{INT}$ and $\mbox{RNE}$ are  almost the same as in the case both  parties have unlimited computational  power. Namely, we show:

\begin{theorem}\label{Th7}
The probabilistic communication complexity for $\mbox{EQ}$ is ${\bf O}(\log n)$ when  parties are using 2PFA.
\end{theorem}
\begin{theorem}\label{Th8}
The quantum communication complexity for $\mbox{INT}$ is ${\bf O}(\sqrt{n}\log n)$ when  parties are using 2QCFA.
\end{theorem}
\begin{theorem}\label{Th9}
The exact quantum communication complexity for $\mbox{RNE}$  is ${\bf O}(n^{0.87}\log n)$ when  parties are using 2QCFA.
\end{theorem}

It seems that for many well known  problems, changing the two communicating parties'computation power to finite automata only does not affect the communication complexity a lot.

 \section{Preliminaries}

   \subsection {Quantum query algorithm}\label{subsec-qqc}

In the following let input $x=x_1\cdots x_n\in\{0,1\}^n$ for some fixed $n$. We will consider a Hilbert space ${\cal H}$ with basis states $|i,j\rangle$ for $i\in\{0,1,\ldots,n\}$ and $j\in\{1,\cdots, m\}$ (where $m$ can be chosen arbitrarily). A query $O_x$ to an input $x\in\{0,1\}^n$ will be formulated as the following unitary transformation:
\begin{itemize}
  \item $O_x|0,j\rangle=|0,j\rangle$;
  \item $O_x|i,j\rangle=(-1)^{x_i}|i,j\rangle$ for $i\in\{1,2,\cdots, n\}$.
\end{itemize}

A  quantum query algorithm ${\cal A}$ which uses $t$ queries for an input $x$ consists of a sequence of
 unitary operators $U_0, O_x, U_1,  \ldots, O_x,U_t$, where $U_i$'s  do not depend on the input
$x$ and the query $O_x$   does. The algorithm will start in a fixed starting state $|\psi_s\rangle$ of ${\cal H}$ and will perform the above sequence of operations. This leads to the final state
\begin{equation}
   |\psi_f\rangle=U_tO_xU_{t-1}\cdots U_1O_xU_0|\psi_s\rangle.
\end{equation}
The final state is then measured with
a measurement $\{M_0, M_1\}$.  For an  input $x\in\{0,1\}^n$, we denote ${\cal A}(x)$  the output of the  quantum query algorithm  ${\cal A}$.  Obviously,
$Pr[{\cal A}(x)=0] =\|M_0|\psi_f\rangle\|^2$ and $Pr[{\cal A}(x)=1] =\|M_1|\psi_f\rangle\|^2=1-Pr[{\cal A}(x)=0]$.
We say that the quantum query algorithm ${\cal A}$ computes $f$ within an
error $\varepsilon$ if for every input $x\in\{0,1\}^n$ it holds that $Pr[{\cal A}(x)=f(x)]\geq 1-\varepsilon$. If $\varepsilon=0$, we says that the  quantum algorithm is an exact quantum algorithm.
For more details on the  definition of quantum query complexity see \cite{Amb13,BdW02}.
\subsection {Communication complexity}

  We will use the following standard
 model of  communication complexity. Two parties Alice and Bob compute a function
$f$ on distributed inputs $x$ and $y$.
A deterministic communication protocol ${\cal P}$ will compute
a  function $f$, if for every  input pair $(x,y)\in X\times Y$ the protocol terminates with the
value $f(x,y)$ as its output at a well specified party.
In a probabilistic  protocol, Alice and Bob may also flip coins during the protocol execution and proceed according to outcomes of the coins. Moreover, the protocol can  have an erroneous output with a small probability.
In a  quantum protocol, Alice and Bob may   use also quantum resources   for communication.
 Let ${\cal P}(x,y)$ denote the  output of the protocol ${\cal P}$. We will consider two  kinds of protocols for computing a function $f$:
\begin{itemize}
  \item An exact protocol ${\cal P}$  such that $Pr({\cal P}(x,y)=f(x,y))=1$.
  \item A  bounded error protocol ${\cal P}$  such that  $Pr({\cal P}(x,y)=f(x,y))\geq \frac{2}{3}$.
\end{itemize}
The communication complexity of a protocol ${\cal P}$  is the number of (qu)bits exchanged in the
worst case.  The communication complexity of $f$ is,   which respect to the communication mode used,  the complexity of an
optimal protocol for $f$.
We will use $D(f)$ and $R(f)$ to denote the deterministic communication complexity and the bounded error probabilistic communication complexity of the function $f$, respectively. Similarly, we use notations $Q_E(f)$ and $Q(f)$ for the exact  and bounded error quantum communication complexity of a function $f$.  For more details on the definition of communication complexity see \cite{Buh09,KusNis97}.


Some communication complexity results that we will use in this paper are:
\begin{enumerate}
  \item  $D(\mbox{EQ})={\bf \Omega}(n)$, $R(\mbox{EQ})= {\bf O}(\log n)$ \cite{KusNis97}.
  \item  $R(\mbox{INT})= {\bf \Omega}(n)$ \cite{Raz92}, $Q(\mbox{INT})= {\bf O}(\sqrt{n}\log n)$ \cite{Buh09}.
  \item  $R(\mbox{RNE})={\bf \Omega}(n)$, $Q_E(\mbox{RNE})={\bf O}({n^{0.87} \log n})$ \cite{Amb13}.
\end{enumerate}

  \subsection {Two-way finite automata}

  We assume familiarity with the models of finite automata introduced in \cite{Amb02,Frd81,HU79}.
  We denote the input alphabet by $\Sigma$, which does not include symbols $|\hspace{-1.5mm}c$ (the left end-marker) and
$\$$ (the right end-marker).
 A two-way finite automaton that we will use in this paper  halts when it enters an accepting or a rejecting state.

2QCFA were introduced by Ambainis and Watrous \cite{Amb02} and further studied  by Zheng \emph{et al.} \cite{GQZ15,LF15,QLMS15,Yak10,ZQLG12,ZQG+13,ZGQ14,ZQG15}. Informally, a 2QCFA can be seen as a 2DFA with an access to a quantum memory for states of a fixed Hilbert space upon which at each step either a unitary operation is performed or a projective measurement and the outcomes of which then probabilistically determine the next move of the underlying 2DFA.

A 2QCFA ${\cal M}$ is specified by a 9-tuple
\begin{equation}
{\cal M}=(Q,S,\Sigma,\Theta,\delta,|q_{0}\rangle,s_{0},S_{acc},S_{rej})
\end{equation}
where:

\begin{enumerate}
\item $Q$ is a finite set of orthonormal quantum basis states.
\item $S$ is a finite set of classical states.
\item $\Sigma$ is a finite alphabet of input symbols and let
$\Sigma'=\Sigma\cup \{|\hspace{-1.5mm}c,\$\}$, where $|\hspace{-1.5mm}c$ will be used as the left end-marker and $\$$ as the right end-marker.
\item $|q_0\rangle\in Q$ is the initial quantum state.
\item $s_0$ is the initial classical state.
\item $S_{acc}\subset S$ and $S_{rej}\subset S$, where $S_{acc}\cap S_{rej}=\emptyset$ are  sets of
the classical accepting and rejecting states, respectively.
\item $\Theta$ is a quantum transition function
\begin{equation}
\Theta: S\setminus(S_{acc}\cup S_{rej})\times \Sigma'\to U(H(Q))\cup O(H(Q)),
\end{equation}
where U(H(Q)) and O(H(Q)) are sets of unitary operations and  measurements on the Hilbert space generated by quantum states from $Q$.

\item $\delta$ is a classical transition function.
If the automaton ${\cal M}$ is in the classical state $s$,  in the  quantum state $|\psi\rangle$, and its tape head is  scanning a symbol $\sigma$, then ${\cal M}$ performs quantum and classical transitions  as follows.
\begin{enumerate}
\item If $\Theta(s,\sigma)\in U(H(Q))$, then the unitary operation $\Theta(s,\sigma)$ is applied on the current quantum state $|\psi\rangle$ to produce a new quantum state. The automaton then performs, in addition, the following classical transition function
\begin{equation}
\delta: S\setminus(S_{acc}\cup S_{rej})\times \Sigma'\to S\times \{-1, 0,1\}.
\end{equation}
If $\delta(s,\sigma)=(s',d)$, then the new classical state of the automaton will be $s'$ and its head moves in the direction $d$.

\item If $\Theta(s,\sigma)\in O(H(Q))$, then the measurement operation $\Theta(s,\sigma)$ is applied on the current state $|\psi\rangle$.
 Suppose the  measurement $\Theta(s,\sigma)$ is specified by operators $\{P_1,\ldots, P_m\}$  and its corresponding classical outcome is from the set $N_{\Theta(s,\sigma)}=\{1,2,\cdots,m\}$.
The classical transition function $\delta$ can be then specified as follow
\begin{equation}
\delta: S\setminus(S_{acc}\cup S_{rej})\times \Sigma'\times N_{\Theta(s,\sigma)}\to S\times \{-1, 0,1\}.
\end{equation}
In such a case,  if $i$ is the classical outcome of the measurement, then the
current quantum state $|\psi\rangle$ is changed to
the  state $P_{i}|\psi\rangle/ \|P_{i}|\psi\rangle\|$. Moreover,  if
$\delta{(s,\sigma)}(i) =(s',d)$, then  the new classical state of the automaton is $s'$ and its head moves in the direction $d$.
\end{enumerate}
The automaton halts and accepts (rejects) the input when it enters a classical accepting (rejecting) state (from $S_{acc}$($S_{rej}$)).

\end{enumerate}

The computation of a 2QCFA
${\cal M}=(Q,S,\Sigma,\Theta,\delta,|q_{0}\rangle,s_{0},S_{acc},S_{rej})$ on an input $w\in \Sigma^*$ starts with the string $|\hspace{-1.5mm}cx\$$ on the input tape. At the start, the tape head of the automation is positioned on the left end-marker and the automaton begins the computation in the classical initial state  $s_0$ and
in the initial quantum state $|q_{0}\rangle$. After that,
in each  step, if  its  classical state  is $s$, its tape head reads a symbol $\sigma$ and its quantum state is $|\psi\rangle$, then the automaton changes its states and makes its head movement following the steps described in the definition.

Let $0\leq\varepsilon<\frac{1}{3}$. A finite automaton ${\cal M}$ recognizes a language $L$ with   error  $\varepsilon$ if, for $w\in \Sigma^*$,
\begin{enumerate}
\item[1.] $\forall w\in L$, $Pr[{\cal M}\  \mbox{accepts}\  w]\geq 1-\varepsilon$, and
\item[2.] $\forall w\notin L$, $Pr[{\cal M}\ \mbox{rejects}\  w]\geq 1-\varepsilon$.
\end{enumerate}

If $\varepsilon=0$, we say the finite automaton ${\cal M}$ is an exact finite automaton.

\section {Proofs}

\begin{proof}[{\bf Proof of  Theorem \ref{Th1}}]

We describe a 2PFA ${\cal A}$ to accept the language $L_{EQ}(n)$. The automaton will use states $s_{q,k,l}$ where $0\leq q,k,l\leq n^2$.

First of all,  ${\cal A}$ uses $O(n)$ states to check that the input is in the form $x\#^ny$, where $|x|=|y|=n$. If the length of the input $|w|>3n$, then the automaton halts and rejects the input in $O(n)$ time.  After that ${\cal A}$ starts an addition computation in the state $s_{0,0,0}$.  After reading the left-end marker, the automaton changes its state randomly to $s_{p,0,0}$, where $p\leq n^2$ is a prime. When the 2PFA ${\cal A}$ reads the ``$x$-region", it changes its state from $s_{p,0,0}$ to $s_{p,s,0}$, where $s=Num(x) \mod p$. $Num(x)$ is the natural number whose binary representation is the string $x$. It is clear that such computation can be done by a 2PFA.  When  ${\cal A}$ reads the ``$\#$-region", it keeps its state unchanged. When  ${\cal A}$ reads the ``$y$-region", it changes its state from $s_{p,s,0}$ to $s_{p,s,t}$, where $t=Num(y) \mod p$. The automaton reaches the right end-marker in a state $s_{p,s,t}$.  If $s=t$, then the input is accepted. If $s\neq t$,  the input is rejected.

 ${\cal A}$ actually simulates the communication protocol \cite{Hrc05} for the problem $\text{EQ}$. If the input $w\in L_{EQ}(n)$, ${\cal A}$ will accept it for certainty.

Let us now say that a prime $2 < p < n^2$ is bad for a pair $(x, y)$ such that $x \neq y$, if the above 2PFA
for such an input pair $(x, y)$ and such a choice of prime yields a wrong answer.
It is clear that there are at most $n-1$ bad primes.
 Let $Prime(m)$ be the number of primes smaller than $m$. By the Prime number theorem,  $Prime(n^2)>\frac{n^2}{2\ln n}$.

If the input $x\#^ny\not\in L_{EQ}(n)$,   ${\cal A}$ accepts the input only with the probability
 \begin{equation}
   \frac{\mbox{number of bad primes}}{Prime(n^2)}<\frac{n-1}{n^2/2\ln n}<\frac{2\ln n}{n}.
 \end{equation}

Obviously, the space used by ${\cal A}$  is   $S={\bf O}(\log n^6)={\bf O}(\log n)$ and the time   is $T={\bf O}( n)$. Therefore,  $TS={\bf O}(n\log n)$.
\end{proof}

\begin{proof}[{\bf Proof of Theorem \ref{Th2}}]

Let ${\cal A}$ be a 2DFA that recognizes the language $L_{EQ}(n)$ in time $T$ using  space $S$. We describe now a deterministic  communication protocol for Alice and Bob that solves the problem  $\text{EQ}$.

 For an input $(x,y)\in\{0,1\}^n\times \{0,1\}^n$, Alice and Bob simulate  ${\cal A}$ with the input $x\#^ny$, where $x,y\in\{0,1\}^n$. It is obvious that  $x\#^ny\in L_{EQ}(n)$ iff $\text{EQ}(x,y)=1$. Alice starts to simulate  ${\cal A}$'s computation as long as the tape head of ${\cal A}$ is either in ``$x$-region" of the input or in the ``$\#$-region" of the input. When the tape head of  ${\cal A}$ moves to the ``$y$-region", then Bob simulates  ${\cal A}$'s computation as long as the tape head of ${\cal A}$ is either in the ``$y$-region" of the input or in the ``$\#$-region" of the input. When the tape head of  ${\cal A}$ moves to the ``$x$-region" of the input, Alice simulates  ${\cal A}$'s computation again. The
idea is that each player is responsible for the simulation in regions where he knows the input bits. In any step in which the tape goes from ``$x$-region" and ``$\#$-region" to ``$y$-region" (from ``$y$-region" and ``$\#$-region" to ``$x$-region"), Alice (Bob) sends the current state of   ${\cal A}$ to Bob (Alice).

In each time, the information which is required to send to the other party is not more than  $S$. Since  move from the ``$x$-region" to the ``$y$-region" and vice versa  takes at least $n$  steps (at least the size of ``$\#$-region"),  the number of times  Alice and Bob send information to each other is at most $T/n$. All together the amount of communicating information in the protocol is not more than $S\cdot T/n$. Since $D(\text{EQ})={\bf \Omega}(n)$ \cite{KusNis97}, we have $S\cdot T/n={\bf \Omega}(n)$  and therefore $TS={\bf \Omega}(n^2)$.
\end{proof}

Before we prove Theorem \ref{Th3}, we present a main proof technique of this paper.  Namely, that every quantum query algorithm can be simulated by a 2QCFA.

\begin{theorem}\label{Query-2QCFA}
The computation of a quantum query algorithm ${\cal A}$ for a Boolean function $f:\{0,1\}^n\to \{0,1\}$ can be simulated by a 2QCFA ${\cal M}$. Moreover, if the quantum query algorithm ${\cal A}$  uses  $t$ queries and $l$ quantum basis states, then the 2QCFA ${\cal M}$ uses  ${\bf O}(l)$ quantum  basis states, ${\bf O}(n^2)$ classical states,  and   ${\bf O}(t\cdot n)$  time.
\end{theorem}

 \begin{proof}
Suppose that the  quantum query algorithm ${\cal A}$ which use $t$ queries is defined as in subsection \ref{subsec-qqc}. The input of the 2QCFA ${\cal M}$  is the same as the input of the  quantum query algorithm ${\cal A}$, which is $|\hspace{-1.5mm}cx\$$ on its tape.  The main idea of the simulation goes as follows:
 We consider now a 2QCFA ${\cal M}$ with quantum basis states $|0\rangle$ and $|i,j\rangle$ for $i\in\{0,1,\ldots,n\}$ and $j\in\{1,\cdots, m\}$. ${\cal M}$ starts its computation in the initial quantum state $|0\rangle$ and  the initial classical state $s_0$.     The first time when ${\cal M}$ reads the left-end marker $|\hspace{-1.5mm}c$,  ${\cal M}$ applies  $\Theta(s_0,\ |\hspace{-1.5mm}c)$ to the quantum state such that $\Theta(s_0,\ |\hspace{-1.5mm}c)|0\rangle=U_0|\psi_s\rangle$.

     The $k$-th time when ${\cal M}$ reads the right-end marker $\$$, ${\cal M}$ applies $U_{k}$ to the quantum state, where $1\leq k\leq t$.
 ${\cal M}$ simulates the query $O_x$ every time when it reads the input $x=x_1\cdots x_n$ from left to right. The automaton proceeds precisely as in Figure \ref{f1}, where
\begin{equation}
\Theta(s_{k,i},\sigma)|0\rangle =|0\rangle, \mbox{ }  \Theta(s_{k,i},\sigma)|i,j\rangle=(-1)^{\sigma}|i,j\rangle\mbox{ and } \Theta(s_{k,i},\sigma)|u,j\rangle=|u,j\rangle \mbox{ for } u\neq i.
\end{equation}
It is easy to verify that the unitary operators preformed in Step 2.1  are
\begin{equation}
    \Theta(s_{k,n},x_n) \Theta(s_{k,n-1},x_{n-1})\ldots \Theta(s_{k,1},x_{1})=O_x.
\end{equation}

 \begin{figure}
\begin{tabular}{|l|}
    \hline
Check that the input $x$ is of the form of $\{0,1\}^n$.
Repeat the following  ad infinity:\\
1. Read the left end-marker $|\hspace{-1.5mm}c$,  perform $\Theta(s_0,\ |\hspace{-1.5mm}c)$ on the initial quantum state $|0\rangle$, change its \\
 \ \ classical state to $\delta(s_0,\ |\hspace{-1.5mm}c )=s_{1,1}$, and move the tape head one cell to the right.\\
2. While the current classical state is not $s_{t+1,1}$, do the following\\
 \ \ 2.1  While the currently  scanned symbol $\sigma$  is not the right end marker $\$$, do the following:\\
 \ \ \ \ 2.1.1 Apply $\Theta(s_{k,i},\sigma)$ to the current quantum state.\\
 \ \ \ \ 2.1.2 Change the classical state $s_{k,i}$ to $s_{k,i+1}$ and move the tape head one cell to the right.\\
 \ \ 2.2 When the right end-marker $\$$  is reached,   perform $\Theta(s_{k,n+1},\$)=U_k$  on the current\\
 \ \ \ \ quantum state. Change the classical state $s_{k,n+1}$ to  $s_{k+1,1}$ and move the tape head  to the\\
 \ \ \ \ symbol  of $x$.\\
3. Measure the current quantum state with the  measurement $\{M_0, M_1\}$.\\
 \ \ If the outcome is 1,  the input is accepted. Otherwise, the input is rejected.\\
    \hline
\end{tabular}
 \centering\caption{Description of the behavior of 2QCFA ${\cal M}$ when simulating the quantum  algorithm ${\cal A}$. }\label{f1}
\end{figure}

It is clear that for any input $x$, $Pr[{\cal A}(x)=1]=Pr[{\cal M} \text{ accepts } x ]$ and $Pr[{\cal A}(x)=0]=Pr[{\cal M} \text{ rejects } x ]$.
 From the above simulation, we can see that if the quantum query algorithm ${\cal A}$  uses $l$ quantum basis states and $t$ queries, then the 2QCFA ${\cal M}$ uses  ${\bf O}(l)$ quantum  basis states, ${\bf O}(n^2)$ classical states,  and    ${\bf O}(t\cdot n)$ time.
\end{proof}

We have proved that 2QCFA can simulate quantum query algorithms. Now what about the quantum communication protocol for the $\text{INT}$ problem? According to \cite{Buh98,Buh09}, we  need to simulate the following unitary map:
\begin{equation}
O_z:|i\rangle \mapsto (-1)^{z_i}|i\rangle,
\end{equation}
where $z = x \wedge y$ is a bit-wise $\mbox{AND}$ of $x$ and $y$, since $z_i = 1$ whenever both $x_i = 1$ and $y_i = 1$.

\begin{lemma}\label{Int-2QCFA}
Let  $w=x\#^ny$, where $x,y\in\{0,1\}^n$, be the input of a 2QCFA ${\cal M}$. Then,   the  unitary map: $O_z:|i\rangle \mapsto (-1)^{z_i}|i\rangle$, where $z = x \wedge y$, can be simulated by  ${\cal M}$. Moreover, ${\cal M}$ uses one additional auxiliary qubit  and ${\bf O}(n)$ classical states and its running time is ${\bf O}(n)$.
\end{lemma}
\begin{proof}

Assume that Alice wants to apply $O_z$ to a quantum state $|\phi\rangle=\sum_{i=1}^n \alpha_i|i\rangle$.  ${\cal M}$ will use quantum states $\{|i\rangle|0\rangle,|i\rangle|1\rangle\}_{i=1}^n$ and classical states $\{s_i\}_{i=0}^{2n+1}$.    ${\cal M}$ will start with the quantum state $|\phi\rangle|0\rangle$.
The procedure to simulate the  unitary map $O_z$ is as in Figure \ref{int}, where
\begin{align}
&  U_{i,\sigma}|j\rangle|b\rangle=|j\rangle|b\oplus \sigma\rangle  \text{ if } j=i, \text{ otherwise }  U_{i,\sigma}|j\rangle|b\rangle=|j\rangle|b \rangle; \\
& V_{i,\sigma}|j\rangle|1\rangle=(-1)^{\sigma}|j\rangle|1\rangle  \text{ if } j=i, \text{ otherwise }\  V_{i,\sigma}|j\rangle|b\rangle=|j\rangle|b\rangle.
\end{align}
It is easy to verify that $U_{i,\sigma}$ and $V_{i,\sigma}$ are unitary.
After Step 2, the quantum state changes to
\begin{equation}
    U_{n,x_n}\cdots  U_{1,x_1}\sum_{i=1}^n \alpha_i|i\rangle|0\rangle=\sum_{i=1}^n \alpha_i|i\rangle|x_i\rangle.
\end{equation}
After Step 4, the quantum state changes to
\begin{equation}
    V_{n,y_n}\cdots  V_{1,y_1}\sum_{i=1}^n \alpha_i|i\rangle|x_i\rangle=\sum_{i=1}^n \alpha_i\cdot(-1)^{x_i\wedge y_i}|i\rangle|x_i\rangle.
\end{equation}
After Step 6, the quantum state changes to
\begin{equation}
  U_{n,x_n}\cdots  U_{1,x_1}\sum_{i=1}^n \alpha_i\cdot(-1)^{x_i\wedge y_i}|i\rangle|x_i\rangle=\sum_{i=1}^n \alpha_i\cdot(-1)^{x_i\wedge y_i}|i\rangle|0\rangle=O_z|\phi\rangle|0\rangle.
\end{equation}

\begin{figure}
\begin{tabular}{|l|}
    \hline
1. Move the tape head to the first symbol of $x$, set its classical state to $s_1$. \\
2. While the currently  scanned symbol $\sigma$ is not $\#$, do the following:\\
\ \ 2.1 Apply $\Theta(s_i,\sigma)=U_{i,\sigma}$ to the current quantum state.\\
\ \ 2.2 Change the classical state $s_i$ to $s_{i+1}$ and move the tape head one cell to the right.\\
3. Move the tape head  to the first symbol of $y$.\\
4. While the currently  scanned symbol $\sigma$ is not $\$$, do the following:\\
\ \ 4.1 Apply $\Theta(s_{n+i},\sigma)=V_{i,\sigma}$ to the current quantum state.\\
\ \  4.2 Change the classical state $s_{n+i}$ to $s_{n+i+1}$ and move the tape head one cell to the right.\\
5. Change the classical state $s_{2n+1}$ to  $s_1$ and move the tape head  to the first symbol of $x$.\\
6. While the currently  scanned symbol $\sigma$ is not $\#$, do the following:\\
\ \ 6.1 Apply $\Theta(s_i,\sigma)=U_{i,\sigma}$ to the current quantum state.\\
\ \ 6.2 Change the classical state $s_i$ to $s_{i+1}$ and move the tape head one cell to the right.\\
\hline
\end{tabular}
 \centering\caption{  Description of the behavior of 2QCFA when simulating the  unitary map $O_z$.}\label{int}
\end{figure}
\end{proof}

\begin{proof}[{\bf Proof of Theorem \ref{Th3}}]
Combining the simulation techniques from Theorem \ref{Query-2QCFA} and  Lemma \ref{Int-2QCFA}, we can use a 2QCFA to simulate a Grover search \cite{Gro96} on the input $z\in\{0,1\}^n$, where $z_i=x_i\wedge y_i$.
Therefore it is clear that there is a 2QCFA recognizing the language
$L_{INT}(n)$.  Since the Grover's algorithm requires ${\bf O}(\sqrt{n})$ queries and uses ${\bf O}(n)$ quantum basis states,  the time used by the 2QCFA is $T={\bf O}(\sqrt{n}\cdot n)={\bf O}(n^{3/2})$. The number of quantum states used by the 2QCFA is ${\bf O}(n)$ and and the number of classical states is ${\bf O}(n^2)$. Therefore, the space used by the 2QCFA is $S={\bf O}(\log n+\log n^2)= {\bf O}(\log n)$. Hence, $TS={\bf O}(n^{3/2}\log n)$.
\end{proof}

\begin{proof}[{\bf Proof of Theorem \ref{Th4}}]
 Similar to the proof of Theorem \ref{Th2} except that probabilistic computation is used  instead of  deterministic  one. The result is based on  $R(\text{INT})= {\bf \Omega}(n)$ \cite{Raz92}.
\end{proof}

\begin{proof}[{\bf Proof of Theorem \ref{Th5}}]
Combining the simulation techniques from Theorem \ref{Query-2QCFA} and Lemma \ref{Int-2QCFA}, we can use a 2QCFA to simulate Ambainis' exact query algorithm in \cite{Amb13} on the input $z\in\{0,1\}^n$, where $z_i=x_i\wedge y_i$. Therefore, there is an exact 2QCFA recognizing the language  $L_{NE}(n)$.  Since the exact algorithm requires ${\bf O}(n^{0.87})$ queries and uses ${\bf O}(n)$ quantum basis states,  the time used by the exact 2QCFA is  $T={\bf O}(n^{0.87}\cdot n)={\bf O}(n^{1.87})$.  The space used is $S={\bf O}(\log n)$. Hence, $TS={\bf O}(n^{1.87}\log n)$.
\end{proof}

\begin{proof}[{\bf Proof of Theorem \ref{Th6}}]
 The proof is similar to that of Theorem \ref{Th2} except that probabilistic computation is used  instead of  deterministic  one. The final result is then based on  $R(\text{RNE})= {\bf \Omega}(n)$ \cite{Amb13}.
\end{proof}

For the proofs of Theorems \ref{Th7},  \ref{Th8} and  \ref{Th9}, it is clear that
the communication protocols can be picked up  from the proofs of Theorems \ref{Th1}, \ref{Th3} and  \ref{Th5}, respectively.  We omit the details of proofs here.




\section{Conclusion and open problems}

Query complexity and communication complexity are related to each other. By using a simulation technique that transforms quantum query algorithms
to quantum communication protocols,  Buhrman et al. \cite{Buh98,Buh09} obtained new quantum communication protocols and showed  the first exponential gap between quantum and classical communication complexity.

In this paper, we have developed the connection among 2QCFA,  quantum communication protocols and quantum query algorithms. We have constructed 2QCFA to simulate quantum query algorithms. Using known quantum query algorithms and quantum communication protocols, this simulation enabled us to prove several time-space tradeoff results for 2QCFA. It also enabled us to find out communication protocols for the case that two parties are using 2QCFA for computation.  It is clear that  if the  protocol is a one-way communication protocol, then the result can be directly transformed to the state complexity result of finite automata. For example, since the protocol in Theorem \ref{Th7} is a one-way protocol,  it is  clear that the following result holds: the space complexity of one-way probability finite automata  for the language $\{x\#y\mid x,y\in\{0,1\}^n, \text{EQ}(x,y)=1\}$ is  ${\bf O}(\log n)$, whereas the space complexity for DFA is ${\bf \Omega}(n)$.
If a one-way quantum communication protocol is transformed from a quantum  query complexity, then it can be implemented on 1QCFA  and the space complexity result will follow immediately. Since we can use known results in quantum query complexity and  communication complexity to derive  new  state succinctness results of quantum finite automata,  the  method is more general than the one used in \cite{Amb98}.

Some  problems for future research:
\begin{enumerate}
  \item The quantum communication complexity tight bound $Q(\text{DISJ})= {\bf \Theta}(\sqrt{n})$ \cite{AA03}. Does there exists a 2QCFA that accepts the language $L_{INT}(n)$ in time $T$ using space $S$ such that
  $TS={\bf O}(n^{3/2})$?

 \item  We have shown,
        for the first time, that the time-space tradeoff $TS$ on exact 2QCFA is superlinearly better than that for 2PFA in recognition of the language $L_{NE}(n)$. Can we find out more languages that 2QCFA have superlinear advantage? Find more examples
such that  exact quantum computing have superlinear advantage in  time-space tradeoff for total functions in other computing models?

  \item We have proved that the exact 2QCFA have superlinear advantage in  time-space tradeoff. Can we prove that exact 2QCFA have  superlinear advantage in time complexity or space complexity in recognizing languages comparing to  2DFA or 2PFA?

 \item We have transformed quantum computing advantages in communication complexity and query complexity to quantum finite automata.  Can we do the opposite way?
  For instant,   Ambainis and Freivalds \cite{Amb98} constructed a quantum finite automaton that is exponentially smaller than equivalent classical automaton, can we transform the problem to communication problem and prove the quantum communication complexity advantage?
\end{enumerate}

\section*{Acknowledgements}
S.G. Zheng  thanks  A.~Ambainis, C.~Mereghetti,  B.~Palano  and L.~Li  for  discussions.

\end{document}